\documentclass{amsart}

\usepackage{amsmath}
\usepackage{amssymb}
\usepackage{amscd}

\newtheorem{theorem}{Theorem}[section]

\newtheorem{definition}{Definition}[section]

\numberwithin{equation}{section}

\begin{document}
	
	\begin{center}
		{\bf All Quantum Probability
			viewed in 
			\\
			Complex Projective Geometry
			}
		
		\vskip 0.2cm
		 Stephen Bruce Sontz
		
		Centro de Investigaci\'on en 
		Matem\'aticas, A.C.
		
		Guanajuato, Mexico

		\vskip 0.3cm 
		In memory of Franciszek Hugon Szafraniec

	\end{center}
	
	\vskip 0.6cm

\section*{Abstract}

In a recent paper it was shown that all the 
Hilbert space 
{\em formulas} for 
quantum probabilities can be realized as
{\em functions} of 
geometric properties of the associated 
projective space, but those functions were
expressed using the structures of the
associated Hilbert space. 
In this paper a direct description 
of all these probabilities is given 
as formulas involving only the geometric properties of
the projective space itself without referring to the
associated Hilbert space theory. 
In large part this depends on a projection theorem
for complex projective space which is analogous to the 
projection theorem for Hilbert spaces.  
The importance 
of this 
is that this 
exhibits quantum probability in terms of the 
geometry of a Riemannian metric in a non-linear 
K\"ahler manifold 
without any reference to a linear 
Hilbert space. 
As such this is a part of a larger program
of the geometrization of physics. 
This opens the possibility of 
generalizations of quantum theory in 
other similar geometric settings. 
The theory presented includes
projective spaces of both finite and infinite dimension. 
Some comments explain how quantum theory based on a
von Neumann algebra is compatible with this 
approach.

\vskip 0.4cm \noindent
Keywords:
Quantum Probability, Complex Projective Space,  Geometrization of
Quantum Theory

\section{Notation and Preliminaries}

This section has a review of known material with some new 
notation. 
The main references are \cite{oneill}, \cite{sbs} and \cite{terek}.
For the initial period of research going up to 1999 
on using complex projective space
as the fundamental setting for quantum theory see the
papers 
\cite{anandan}, \cite{ashtekar-schilling}, \cite{cirelli1},
\cite{cirelli2},  \cite{gibbons}, \cite{heslot},
\cite{hughston}, \cite{kibble}, \cite{manin},
\cite{marsden-ratiu} and \cite{schilling}.  
A more recent and 
extensive overview is \cite{freed}, 
which is based on lectures given in 2023. 

The probability functions of quantum theory
are initially given in the Hilbert space 
setting as explicit formulas in \cite{sbs}. 
In order to find the corresponding explicit formulas in the 
projective space setting, it is first necessary 
to review the Hilbert space
setting.
So from now on $\mathcal{H}$ denotes 
a complex Hilbert space, not necessarily
separable. 
But to 
avoid trivial cases we always assume that
$\dim_{\mathbb{C}} \mathcal{H} \ge 2$. 
The notation for the norm on $\mathcal{H}$ 
is $|| \cdot ||$ and for the Hermitian inner product it
is $\langle \cdot , \cdot \rangle$ with the 
convention that this is conjugate linear in 
the first entry and linear in the second. 
We use the notation 
$$
\mathcal{L} (\mathcal{H}) = 
\{ T : \mathcal{H} \to \mathcal{H} ~|~
T {\rm ~is~linear~and~bounded} \},
$$ which is
a unital $C^{*}$-algebra with respect to the
operator norm and conjugation given by 
taking the adjoint. 
It is also a Type~I factor. 

Next, for any vector space $V$ we put
$V^{\rm x}:= V \setminus \{ 0\}$. 
Then the {\em complex projective space} 
associated to $\mathcal{H}$ is defined as 
$$
  \mathbb{CP} (\mathcal{H}) := 
  \mathcal{H}^{\rm x} / \mathbb{C}^{\rm x},
$$
where we quotient by the multiplicative action
of the non-zero 
scalars $ \mathbb{C}^{\rm x} $ on
the set $ \mathcal{H}^{\rm x}$. 
Since $\mathcal{H}$ has a norm, this is 
equivalent to 
$ \mathbb{CP} (\mathcal{H}) := \mathbb{S}(\mathcal{H})/S^{1}$, 
where 
$$
\mathbb{S}(\mathcal{H}) := \{ \psi \in \mathcal{H}
~|~ || \psi || =1  \} \qquad \mathrm{and} \quad 
S^{1}:= \mathbb{S}(\mathbb{C}) = 
 \{ z \in \mathbb{C} ~|~ |z| =1\}
$$
are called the unit sphere in $\mathcal{H}$ and the unit circle,
respectively. 

Let $\pi_{1} : \mathcal{H}^{\rm x} \to \mathbb{S}(\mathcal{H})
$ be the quotient map $\pi_{1} (\psi) := \psi / || \psi ||  $ 
for all $\psi \in \mathcal{H}^{\rm x} $. 

Let 
$ \pi_{2} : \mathbb{S}(\mathcal{H}) \to 
\mathbb{CP} (\mathcal{H})$
denote the quotient map sending a unit vector to its 
equivalence class modulo the action of $S^{1}$. 

Finally, let $\pi_{3} := \pi_{2} \pi_{1} : 
\mathcal{H}^{\rm x} \to \mathbb{CP} (\mathcal{H})$. 

It is known that $ \mathbb{CP} (\mathcal{H})$
is a K\"ahler manifold, including the important 
case when $\mathcal{H}$ is infinite dimensional, provided 
one uses a natural generalization of the finite
dimensional definitions 
to that setting. 
(See \cite{schilling}.)
Recall that a K\"ahler manifold is a complex 
manifold with compatible almost complex structure,
symplectic form and Riemannian metric. 
In the case of  $ \mathbb{CP}  (\mathcal{H})$
this Riemannian metric is called the
{\em Fubini-Study metric}, whose associated metric on the 
manifold,
defined as the infimum of  the lengths of $C^{1}$ curves, 
is denoted as 
$$
 d(x,y) \quad {\rm for~} x,y \in \mathbb{CP}  (\mathcal{H}).
$$

The Fubini-Study metric can be defined in terms of the
standard Riemannian metric on $\mathbb{S(\mathcal{H})}$ 
(induced from its embedding in $\mathcal{H}$) 
by using $\pi_{2}$. 
Then {\em horizontal tangent vectors} in the tangent bundle
$T(\mathbb{S(\mathcal{H})})$ are defined to
be the elements in $(\ker T \pi_{2})^{\perp} $. 
Also a curve in $ \mathbb{S(\mathcal{H})}$ is said to be
{\em horizontal} if all of its tangent vectors are horizontal. 
Every curve $\gamma_{C}$ in $\mathbb{CP}  (\mathcal{H})$ has a unique {\em horizontal lift} $\gamma_{S}$ in $\mathbb{S(\mathcal{H})}$, 
that is, $\gamma_{C}  =  \pi_{2} \circ \gamma_{S}$ and
$\gamma_{S}$ is a horizontal curve  
in $\mathbb{S(\mathcal{H})}$. 
This way $\pi_{2}$  induces a bijection between the
horizontal geodesics of $\mathbb{S(\mathcal{H})}$ and all
the geodesics of $ \mathbb{CP}  (\mathcal{H})$.
Also the length of a horizontal curve
in $\mathbb{S(\mathcal{H})}$ is equal to the length of 
its image under $\pi_{2}$ in $ \mathbb{CP}  (\mathcal{H})$. 
This geometric framework for quantum probability is based
on the general results that hold for $\pi_{2}$ as a
Riemannian submersion together with the specific 
geometries of $ \mathbb{S(\mathcal{H})}$ and
$ \mathbb{CP}  (\mathcal{H})$. 
See \cite{oneill} and \cite{terek} for details. 

In terms of geometry some important facts are that
the diameter of $ \mathbb{S(\mathcal{H})}$ is $\pi$ while that of
$ \mathbb{CP}  (\mathcal{H})$ is $\pi/2$. 
The diameter is defined in terms of the 
Riemannian metric and so is independent of 
any metric induced by an embedding in another space. 
Also, both spaces have the property that if two points
have a distance strictly less than their diameter, then
they can be connected by a unique shortest geodesic. 
On the other hand if two points have a distance equal to the
diameter, then they can be connected  with infinitely 
many shortest geodesics. 

The isomorphic lattice structures of these spaces 
in terms of their subspaces in the appropriate category 
plays as well an essential role in quantum probability 
theory, though that aspect is not emphasized here. 

The (geometric!) points of the projective space 
$ \mathbb{CP}  (\mathcal{H})$ 
are also called the {\em (pure) states} of the
quantum system represented by $\mathcal{H}$. 
However, we often say that an element in
$\mathcal{H}^{\rm x}$
or in $ \mathbb{S}(\mathcal{H})$ 
{\em represents} the state given by its image under
the mapping $\pi_{3}$ or $\pi_{2}$, respectively. 

It is important to note that $\mathcal{H}$ is also a
vector space over the real numbers $\mathbb{R}$ and that
it becomes a real Hilbert space with the inner product
\begin{equation}
\label{real-inner-product}
    \langle \phi, \psi \rangle_{\mathbb{R}} :=
   \Re \langle \phi, \psi \rangle,	
\end{equation}
where $\Re$ denotes the real part. 
Notice that the associated norm of this real inner product
coincides with that of the complex inner product, since
\begin{equation*}
	\langle \phi, \phi \rangle_{\mathbb{R}} =
	\Re \langle \phi, \phi \rangle = 
	\langle \phi, \phi \rangle = || \phi ||^{2} \quad 
	{\rm for~all~} \phi \in \mathcal{H}. 
\end{equation*}
For example, the definition above of 
$\mathbb{S}(\mathcal{H})$ gives exactly the same 
the unit sphere
with respect to either definition of the norm. 
Hereafter, we will use only the notation $|| \cdot ||$ 
for the norm. 
Actually, $\mathcal{H}$ is a K\"ahler manifold whose
symplectic form is the imaginary part of the 
Hermitian inner product.

The {\em projective angle} $a$ between vectors 
$\phi,\psi \in \mathcal{H}^{\rm x}$ 
is defined to be
$$
    a:= \cos^{-1} 
\left( \dfrac{|\langle \phi, \psi \rangle|}{|| \phi || \, ||\psi||} \right)
$$
where we take the branch of the inverse cosine satisfying 
$\cos^{-1} : [0, 1] \to [0,\pi/2]$. 
Note that the projective angle passes down to the 
projective space, justifying the terminology. 

The {\em Euclidean angle} $\theta$ between vectors 
$\phi,\psi \in \mathcal{H}^{\rm x}$ is defined to be
$$
\theta:= \cos^{-1} 
\left( \dfrac{\langle \phi, \psi \rangle_{\mathbb{R}}}
{|| \phi || \, ||\psi||} \right)
$$
where we take the branch of the inverse cosine satisfying 
$\cos^{-1} : [-1, 1] \to [0,\pi]$. 
The Euclidean angle is what is used in Euclidean geometry.
However, it does not pass down to the projective space. 

Notice that these two definitions of angle coincide 
if and only if $ \langle \phi, \psi \rangle \ge 0 $.

Moreover 
notice that the inner product $\langle \psi, \phi \rangle$
for unit vectors $\psi, \phi$ does not pass 
to the projective space, while 
$|\langle \psi, \phi \rangle|$ does. 
So for $x, y \in \mathbb{CP} (\mathcal{H})$ we define 
the {\em absolute inner product}
$$ 
|\langle x, y \rangle|:= |\langle \psi, \phi \rangle| 
$$
for any choice of $\psi \in \pi_{2}^{-1}(x)$ and
$\phi \in \pi_{2}^{-1}(y)$. 
(Beware that neither this terminology nor 
this notation implies that
$ \langle x, y \rangle$ exists and that it has an
absolute value.)
Another way to write this definition is
$$ 
|\langle \pi_{2}(\psi), \pi_{2}(\phi) \rangle|:= |\langle \psi, \phi \rangle| 
$$
for all unit vectors $\psi, \phi \in \mathcal{H}$.
Notice that $|\langle x , x \rangle| = 1 $ for
all $x \in \mathbb{CP} (\mathcal{H})$. 

Next we define $x,y \in \mathbb{CP} (\mathcal{H})$ 
to be {\em orthogonal}
if $ |\langle x, y \rangle| = 0 $ in which
case we denote this as $x \perp y$. 
Clearly, $x \perp y$ if and only if $y \perp x$. 
For any subset $X \subset \mathbb{CP} (\mathcal{H})$ 
we define its {\em orthogonal complement} by
$$
X^{\perp}:= \{  y \in \mathbb{CP} (\mathcal{H}) ~|~
x \perp y {\rm ~for~all~} x \in X \}. 
$$
It follows that $X \cap X^{\perp} = \emptyset$, the
empty set, since assuming that there exists some
$x \in X \cap X^{\perp}$ it follows that $x \perp x$ 
or, equivalently, $|\langle x , x \rangle | = 0$ which 
contradicts $|\langle x , x \rangle | = 1$. 

The absolute inner product is well defined as
a function from $\mathbb{CP} (\mathcal{H}) \times \mathbb{CP} (\mathcal{H})$
to $[0,1]$ and is invariant under the projective unitary group.
As such it must be considered as a geometric 
property. 
However, the goal here is to realize this and other related
functions as {\em formulas} with respect to the K\"ahler space 
structures of the complex analytic manifold 
$\mathbb{CP} (\mathcal{H})$. 

Also the unit sphere $\mathbb{S}(\mathcal{H})$ is an
embedded real smooth (that is, $C^{\infty}$) 
submanifold of the real smooth manifold 
$\mathcal{H}$. 
Actually, $\mathbb{S}(\mathcal{H})$ is the total space
of a principal fiber bundle whose real Lie group
is $S^{1}$ and whose base space is 
$ \mathbb{CP}  (\mathcal{H})$. 
The three spaces of this bundle are $C^{\infty}$ manifolds
given the appropriate definitions for the important 
infinite dimensional case.
A curious topological fact, 
which will not be used here,
is that for a separable, infinite dimensional
Hilbert space we have that 
$ \mathbb{CP}  (\mathcal{H})$ is an
Eilenberg-Mac Lane space of type $K(\mathbb{Z},2)$. 

The results presented here are consequences 
mainly of the 
smooth Riemannian structure of these
manifolds. 
More precisely, 
the  symplectic structure of some of these 
manifolds is an incidental feature that is not used here. 
In particular, $\mathbb{S(\mathcal{H})}$ 
typically does not have a symplectic structure,
though it always has a Riemannian metric. 
The results also depend on the absolute value 
of the Hermitian inner product, which we will
see is a function of a geometric property of 
$ \mathbb{CP}  (\mathcal{H})$. 
However, the Hermitian inner product on $\mathcal{H}$ does not pass down to the 
projective space. 

We recall that an {\em event} of 
a Hilbert space $\mathcal{H}$ is defined to be
any $E \in \mathcal{L} (\mathcal{H})$ 
satisfying $E = E^{2} = E^{*}$ or, equivalently
in the language of Hilbert space theory, $E$ is
an {\em (orthogonal) projection}. 
It bears mentioning that the $^{*}$ in 
$E^{*}$ refers to the adjoint with
respect to the Hermitian inner product 
$\langle \cdot, \cdot \rangle $ of $\mathcal{H}$. 

The set of events of $\mathcal{H}$ is in bijection 
with the set of closed subspaces of $\mathcal{H}$
via $E \mapsto {\rm Ran} \, E$, the range of $E$. 
Also, the set 
of closed subspaces of $\mathcal{H}$ is in bijection with the set of (complex) projective
subspaces of $\mathbb{CP} (\mathcal{H})$
via $V \mapsto \pi_{3}(V^{\rm x})$ for 
all closed complex subspaces $V$ of $\mathcal{H}$. 
This last statement 
is in some narrow sense simply a matter of {\em logical}
definition, because the projective
subspaces of $\mathbb{CP} (\mathcal{H})$
can be defined to be those subsets of the form  
$\pi_{3}(V^{\rm x})$ for some (actually unique) 
closed complex vector subspace $V$ of $\mathcal{H}$.
However, this is not a good 
{\em geometrical}
definition, which
should be given only in terms of the 
intrinsic\footnote{This is in resonance with the 
Theorema Egregium of Gauss, namely that 
curvature is an intrinsic geometric property 
of a surface and not of its embedding in another space.}
properties of $\mathbb{CP} (\mathcal{H})$ without any
reference to other external structures, which in this 
case are those of the Hilbert space $\mathcal{H}$. 
At least in the finite dimensional case, 
the appropriate intrinsic structures in this situation are
the topology,
the Riemannian metric and the complex manifold structure,
since the complex projective
subspaces of $\mathbb{CP} (\mathcal{H})$ are precisely 
the connected, 
totally geodesic, complex submanifolds and therefore this
combination of properties can be taken as a good intrinsic 
geometric definition. 
A corollary of this viewpoint is that the notation 
$\mathbb{CP} (\mathcal{H})$ should be changed to 
something without any reference at all to 
the Hilbert space $\mathcal{H}$, though that shall not
be done here. 

To simplify notation we define
$\pi' (E) := \pi_{3} ( ( {\rm Ran} \, E)^{\rm x})$. 
Then $\pi'$ is a bijection from the set of events to the
set of complex projective subspaces. 
In particular, for $E=0$ we have 
$\pi' (0) = \emptyset $, the empty projective subspace
while for $E=I$, the identity, we have 
$\pi' (I) = \mathbb{CP} (\mathcal{H})$. 

Here are some useful properties of geodesics in 
$ \mathbb{S}(\mathcal{H})$ and $ \mathbb{CP} (\mathcal{H})$.
These expand on what was presented earlier. 
(See \cite{manin}, \cite{oneill}.)
	\begin{itemize}
		\item 
		Between any pair of points
		$\phi,\psi \in \mathbb{S(\mathcal{H})}$ there
		is a shortest geodesic with length 
		$\theta \in [0, \pi]$, 
		where $\theta = \cos^{-1} \langle \phi, \psi \rangle_{\mathbb{R}}$ is the
		Euclidean angle between 
		$\phi$ and $ \psi$. 
		This geodesic is unique if and only if 
		$\theta < \pi$. 
		
		\item 
		Between any pair of points
		$ x,y \in \mathbb{CP} (\mathcal{H})$ there
		is a shortest 
		geodesic with length 
		$a \in [0, \pi/2]$, 
		where $a = \cos^{-1} |\langle x,y \rangle|$ is
		the projective angle between $x$ and $y$. 
		This geodesic is unique if and only if $a < \pi/2$. 
	\end{itemize}

The rest of the paper is organized as follows. 
In Section 2 we review a previously known special
case of our results in order to establish some
background and to provide motivation. 
Then in Section 3 there is a Projection Theorem,
which is a key geometric property of complex projective
spaces. 
In section \ref{one-event} the probability of
one event is given in geometric terms, while 
Section 5 gives the consecutive probability of
two time ordered events in geometric terms. 
Section 6 describes how this approach is 
compatible with quantum theory based on
a von Neumann algebra. 
Finally, there is a section with some concluding remarks
and possible generalizations of this research.

\section{A Special Case}

One special case of the goal of this 
paper is already known. 
(See \cite{ashtekar-schilling}, \cite{schilling}.)
This is the case 
of the Born Rule for the transition probability 
as expressed in a Hilbert space setting as the
probability of finding a system in a state represented by 
a unit vector $\psi \in \mathcal{H}$ given
that it currently is in  a
state represented by 
a unit vector $\phi \in \mathcal{H}$. 
This probability is given by 
$$
{\rm P_{Hilb}} (\psi | \phi) := 
|\langle \psi, \phi \rangle|^{2}. 
$$
To see that 
this function passes to the projective space 
one puts $x:= \pi_{2}(\psi')$ and $y:= \pi_{2}(\phi')$
and verifies that the probability of $x$ 
given $y$ is
\begin{equation}
	\label{x-given-y}
{\rm P_{Proj} } (x|y) := 
|\langle \psi', \phi' \rangle|^{2} = 
{\rm P_{Hilb}} (\psi' | \phi')
\end{equation}
is well defined for any choice of 
$\psi' \in \pi_{2}^{-1}(x)$ and of 
$\phi' \in \pi_{2}^{-1}(y)$.

Given this notation we can rewrite 
\eqref{x-given-y} as 
\begin{equation*}
	{\rm P_{Proj} } (x|y) := 
	|\langle x, y \rangle|^{2}. 
\end{equation*}

Nonetheless the formula on the right side
of \eqref{x-given-y} for evaluating the
function $ {\rm P_{Proj}}$ depends on structures 
of the Hilbert space $\mathcal{H}$. 
However, one has the following result 
\begin{equation}
\label{projective-probability}
{\rm P_{Proj}  } (x|y) = 
|\langle x, y \rangle|^{2} = 
\cos^{2} d(x,y), 
\end{equation}
where $d(x,y) \in [0,\pi/2]$ is the length of a
(not necessarily unique) shortest geodesic 
with respect to the Fubini-Study metric
from $x$ to $y$ in 
$\mathbb{CP} (\mathcal{H})$ as is proved in 
\cite{ashtekar-schilling} and \cite{schilling}. 
This fulfills the goal of expressing this
probability in purely geometric terms of 
$\mathbb{CP} (\mathcal{H})$. 

This previously known result is a a special
case of the result proved in Section~\ref{one-event}. 
As an added benefit \eqref{projective-probability}
gives us a formula for the absolute inner
product in terms of geometric properties of 
$\mathbb{CP} (\mathcal{H})$, since taking the 
positive square root of both sides again 
produces an equality because of the range
of $d(x,y)$.

\section{A Projection Theorem}

We will now see an important geometric fact 
about complex projective spaces. 
This result is analogous to as well as being based on
the Projection Theorem  of Hilbert space theory. 
First, here is a definition. 

\begin{definition}
Suppose that $Y \subset \mathbb{CP} (\mathcal{H})$ is any 
non-empty subset and $x \in \mathbb{CP} (\mathcal{H})$. 
Then the {\em distance from $x$ to $Y$} is defined to be 
$d(x,Y):=\inf \{ d (x,y)~|~ y \in Y\}$.
\end{definition}

Before stating the projection 
theorem it may be appropriate to 
remark that the empty subset $\emptyset$ of 
$ \mathbb{CP} (\mathcal{H})$ is 
a projective subspace whose dimension is $-1$. 
Also, since $\mathcal{H}$ can have infinite dimension, 
it (as well as some of its projective subspaces) could
fail to be compact. 
Now here is the theorem itself.

\begin{theorem}
{\bf Projection Theorem of Projective Space:}
Let $S \subset \mathbb{CP} (\mathcal{H})$ be a 
non-empty projective
subspace and $x \in \mathbb{CP} (\mathcal{H})$. 
Then there exists a point $y \in S$ such that
$d(x,S) =  d (x,y)$, which is equal to 
the projective angle 
between $x$ and $y$. 
There are two mutually exclusive cases.
\begin{itemize}
	
\item 
If $x \notin S^{\perp}$, then $y$ is unique 
and $y = \pi_{3}(E\psi)$, where 
$E= \pi'^{-1}(S)$, 
for any $\psi \in \mathbb{S(\mathcal{H})}$ 
satisfying $\pi_{2}(\psi) = x$. 
Moreover, $ 0 \le d(x,y) < \pi/2$. 
In this case we define the {\em projection of $x$ on $S$} 
to be ${\rm Pj}(x \downarrow S) :=y$. 

\item 
If $x \in S^{\perp}$, then $d (x,S) = d(x,y) = \pi/2$ 
for all $y \in S$. 
In this case $y$ is unique if and only if 
$\dim S = 0$, that is, $S$ is a single point. 
We also define the {\em projection of $x$ on $S$} 
to be ${\rm Pj}(x \downarrow S) :=S$. 
\end{itemize}

\end{theorem}
\noindent
{\bf Remark:}
It is important to note that 
the dichotomy used for these two cases is 
a geometric property of 
$\mathbb{CP} (\mathcal{H})$, since it is based 
on the absolute inner product, which itself is
a geometric property. 

\noindent
{\bf Proof:}
There is a unique event $E \in \mathcal{L} (\mathcal{H})$ such 
that $\pi'(E) = S$. 
Also, there exists 
 $\psi \in \mathbb{S}(\mathcal{H})$ 
such that $\pi_{2} (\psi) = x$. 
Note that $\psi $ is not unique; 
this fact will be taken into consideration 
in the following argument.

The result is verified for $x \in S$, since then 
$y = x$ is clearly the unique point sought for and moreover
$d(x,y)=0= d(x,S)$. 
So for the remainder of this proof we will assume that
$x \notin S$, which is equivalent to $\psi \ne E\psi$ 
and also equivalent to $\psi \notin {\rm Ran} \, E$. 

Here is the proof of the first assertion. 
Notice that $x \notin S^{\perp^{}}$ is equivalent 
to $E \psi \ne 0$, that is, 
$\psi \notin \ker E$. 
In particular, $\psi$ and $E \psi$ are linearly independent 
over $\mathbb{C}$. 
Then the line segment $g$ from $\psi $ to $E \psi$ lies
entirely in $\mathcal{H}^{\rm x}$. 
To verify this we note that otherwise 
$(1-t) \psi + t E \psi = 0$ for some $t\in (0,1)$ implying 
that the non-zero vector $\psi $ is 
an eigenvector of $E$, which in turn 
implies either $\psi \in {\rm Ran} \, E$ or 
$\psi \in \ker E$. 
But these two possibilities have been excluded. 

The remainder of the argument is largely 
geometry in the style of Euclid. 
First, by the Projection Theorem for Hilbert space 
(see \cite{kato}, Chap. 5)
we have that the line segment 
$g$, defined above,
is (or more precisely is the image of) 
a minimal geodesic in the metric of $\mathcal{H}$
from $\psi$ to $ {\rm Ran} \, E $ and that 
$g$ forms a right angle with the vector
$E \psi$ with respect to both inner products. 
Then $g$ and $E\psi$ are the short sides of a right triangle 
$ T$
whose hypotenuse is $\psi$. 

Notice that the Hermitian inner product
of the side $E \psi$ and the hypotenuse
$\psi$ satisfies 
$\langle \psi , E \psi \rangle = || E \psi ||^{2} \in (0,1)$. 
(The end points of this interval are excluded because of
the case we are considering.)
Consequently, the Euclidean angle $\theta$
formed by the side $E \psi$ and the hypotenuse
$\psi$ is acute. 
Following a standard abuse in notation we denote an angle
and its size in radians with the same symbol. 
Accordingly, we have that $\theta \in (0,\pi/2)$. 

Next, since $ \psi, E\psi$ are linearly independent
(over $\mathbb{C}$ and hence also over $\mathbb{R}$), 
the real vector space spanned by $\psi$ and $E\psi$ is
a real 2-dimensional plane, which 
contains the triangle $T$, defined above, 
 and which intersects 
$\mathbb{S}(\mathcal{H})$ in a great circle $C_{1}$. 
Then $\pi_{1}$ maps the points on $g$ to the points 
on the great circle $C_{1}$ that lie between $\psi$ and
$\pi_{1}(E\psi)$. 
(In this context `between' means on the shortest geodesic 
in $ \mathbb{S}(\mathcal{H}) $
connecting $\psi$ and $\pi_{1}(E\psi)$.)
In other words $\pi_{1}(g)$ is the image of the shortest 
geodesic $g'$ in $ \mathbb{S}(\mathcal{H}) $ 
from $\psi$ to $\pi_{1}(E\psi)$. 
The length of this geodesic in the real manifold 
$\mathbb{S}(\mathcal{H}) $ 
is precisely the size in radians of the acute 
Euclidean angle formed by 
its endpoints $\psi$ and $\pi_{1}(E\psi) = E\psi/ || E\psi ||$
or, equivalently, by the acute Euclidean angle $\theta$ 
formed by $\psi$ and $E\psi$. 

However, the argument up to this point does depend on
the choice of $\psi$ such that $x = \pi_{2}(\psi)$. 
It also remains to pass these results down to the
quotient projective space. 
First, note that replacing $\psi$ with $e^{i\beta}\psi$
with $\beta$ real in 
the argument above changes the point 
$\pi_{1}(E\psi) = E\psi/ || E\psi ||$ in $\mathbb{S(\mathcal{H})}$
to 
$$
\pi_{1} ( E(e^{i\beta}\psi) ) = 
E(e^{i\beta}\psi)/ || E(e^{i\beta}\psi) || = 
e^{i\beta}E\psi/ |e^{i\beta}| \, || E\psi || =
e^{i\beta} \pi_{1}(E\psi) 
$$
and therefore
$$
\pi_{3} ( E(e^{i\beta}\psi) ) =
\pi_{2} ( \pi_{1} ( E(e^{i\beta}\psi) )  ) = 
\pi_{2} (  e^{i\beta} \pi_{1}(E\psi) )  = 
\pi_{2} (  \pi_{1}(E\psi) ) =
\pi_{3} (E\psi),
$$
which means that the corresponding point in 
$\mathbb{CP} (\mathcal{H})$ 
does not change under this replacement. 
We now have to show that $y := \pi_{3} (E\psi)$ 
is the point we are seeking. 

The claim is that every other curve in 
$\mathbb{CP} (\mathcal{H})$ from $x$ to $S$
has length strictly greater than $\theta$. 
So consider any 
curve from $x = \pi_{2}(\psi)$ to some
$z \in S$ such that $z \ne \pi_{3} (E \psi) = y$, 
where, as noted above, 
$ y = \pi_{3} (E \psi)$ does not depend on the choice 
of $\psi$. 
We must show this curve is strictly longer. 
It suffices to take the shortest such curve, namely 
the shortest geodesic $\tilde{g}$ in 
$\mathbb{CP} (\mathcal{H})$
from $x$ to $z$. 
Recall that $x \notin S$. 
So then $z \in S$ implies that $x \ne z$. 

Next, we consider the unique horizontal lift $h$ of $\tilde{g}$
to $\mathbb{S} (\mathcal{H})$ 
that starts at $\psi $.
So $h $ is a horizontal geodesic in $\mathbb{S} (\mathcal{H})$
(see Cor. 33 in \cite{terek})
that starts at $\psi$ and terminates at some 
$\phi \in {\rm Ran }\, E \cap \mathbb{S} (\mathcal{H})$ 
that satisfies 
$\pi_{2}(\phi) = z \ne x = \pi_{2}(\psi)$. 
Therefore, $\phi \ne \pm \psi$. 
It follows that the three distinct 
vectors $\pi_{1}(E\psi), \phi, \psi$
are linearly independent
as before over $\mathbb{C}$ and hence over $\mathbb{R}$. 
To see that this is so, note that any linear 
combination of these being equal to zero implies
that the coefficient of $\psi$ is zero,
since $ \pi_{1}(E\psi)$ and $ \phi $ lie
in the subspace ${\rm Ran} \, E$ while $\psi$ 
does not. 
So this linear combination now only includes
the two distinct unit vectors
$ \pi_{1}(E\psi)$ and $ \phi $ and so they 
are linearly independent. 

The argument reduces to analyzing 
 the geometry of a two
dimensional sphere. 
This sphere, denoted by $\Sigma'$, is the
unit sphere in the real three dimensional 
vector space $V$ spanned by the 
unit vectors $\pi_{1}(E\psi), \phi, \psi$. 

Let $\theta'$ denote the Euclidean 
angle between $\phi$ and $\psi$. 
Thus $\cos \theta' = \langle \phi, \psi \rangle_{\mathbb{R}}$  
and, since $ \phi \ne \pm \psi$ we have $\theta' \in (0,\pi)$. 
Recall that $\theta \in (0,\pi/2)$ is the Euclidean angle 
between $\pi_{1}(E\psi) $ and $ \psi$ and also is
the distance between $\pi_{3}(E\psi) $ and $x = \pi_{2}(\psi)$
in the Fubini-Study metric of 
$\mathbb{CP} (\mathcal{H})$. 

Now we pick an orthonormal basis of $V$ with respect to the
inner product $\langle \cdot , \cdot \rangle_{\mathbb{R}}$
such that its Cartesian coordinates satisfy
${\rm Ran} \, E = \{ (r,s,0)  ~|~ r,s \in \mathbb{R}\}$. 
Then
$\psi$ has real coordinates $(x_{1}, x_{2}, x_{3})$ with $x_{3} \ne 0$ 
and $x_{1}^{2} + x_{2}^{2} + x_{3}^{2} = 1$ and $E \psi$
has coordinates $(x_{1}, x_{2}, 0) \ne (0,0,0)$.  
Finally, $\phi$ has real coordinates $(a,b,0)$ such that 
$ a^{2} + b^{2} = 1 $. 
It follows that 
$$
\cos \theta = \langle \psi , \pi_{1}(E\psi) \rangle_{\mathbb{R}} =
\langle (x_{1}, x_{2}, x_{3}), (x_{1}, x_{2}, 0) \rangle_{\mathbb{R}} =
(x_{1}^{2} + x_{2}^{2})^{1/2} > 0.  
$$
On the other hand
$$
|\cos \theta' | = | \langle \phi, \psi \rangle_{\mathbb{R}} | = 
| a x_{1} + b x_{2} | < 
(a^{2} + b^{2})^{1/2} (x_{1}^{2} + x_{2}^{2})^{1/2} =
\cos \theta. 
$$
The strict form of the Cauchy-Schwarz inequality applies
since $\phi $ and $ \psi$ are linearly independent. 
Now we consider two cases. 

In the first case $\cos \theta' \ge 0$, which is equivalent
to $\theta' \in [0, \pi/2]$. 
Then we have $\cos \theta' < \cos \theta$, which given 
the allowed ranges of $\theta$ and $\theta'$ implies that
$\theta' > \theta > 0$. 
Now $\theta'$ is the distance between $\phi$ and $\psi$ 
which turns out also to be the distance between 
$\pi_{2}(\phi)$ and $x = \pi_{2}(\psi)$ in 
$\mathbb{CP} (\mathcal{H})$ because $\theta' \in [0, \pi/2]$. 
And so we have proved that the distance  $\theta'$ from 
$\pi_{2}(\phi)$ to $x$ is strictly greater than 
the distance $\theta$ from 
$y = \pi_{3}(E\psi)$ to $x$,
as desired. 

In the second case $\cos \theta' < 0$, which is equivalent
to $\theta' \in (\pi/2, \pi]$. 
Then we have $- \cos \theta' < \cos \theta$. 
Given the allowed ranges of $\theta$ and $\theta'$ this 
implies\footnote{To see this quickly just graph
$\cos$ on $[0, \pi]$ and note what the last inequality implies about the distances of 
$\theta$ and $\theta'$ from $\pi/2$.} 
$\theta' - \pi/2 < \pi/2 - \theta $, 
or equivalently $\theta < \pi - \theta' \in [0, \pi/2]$. 
Now the shortest geodesic from $\psi$ to $-\phi$ has length 
equal to the angle between $\psi$ to $-\phi$, which is
$ \pi - \theta' $, 
the supplementary angle of $\theta'$. 
This shortest geodesic lies on the same great circle as 
$h$ and so is horizontal.
Therefore, it 
projects down to a geodesic of the
same length $\pi - \theta'$ between $\pi_{2}(\psi)$ and 
$\pi_{2}(-\phi) = \pi_{2}(\phi) = x$. 
And this distance is strictly greater than 
the distance $\theta$ from 
$y = \pi_{3}(E\psi)$ to $x$,
again as desired. 
This finishes the proof of the first assertion.

For the second assertion, note that $x \in S^{\perp^{}}$ just 
means that $\psi$ is orthogonal to 
the `equator' $ ({\rm Ran} \, E) \cap \mathbb{S}(\mathcal{H}) $
or, in other words, the minimal geodesic from $\psi$ 
to any point on the `equator' has length $\pi/2$, which is just
the magnitude in radians of the angle 
between these two points. 
Clearly, the point  $y \in S$ is unique if and only if the
non-empty subspace 
$S$ itself has exactly one point. 
$\quad \blacksquare$

\vskip 0.2cm 
As a further comment on the second assertion, notice 
in that case that the geodesic between $\psi$ and $E \psi = 0$
in $\mathcal{H}$ is $[0,1] \ni t \mapsto t \psi$. 
Moreover, this geodesic lies in $\mathcal{H}^{\rm x}$ for
$t \ne 0  $ and its image in  $\mathbb{CP} (\mathcal{H})$ 
under the map $\pi_{3}$
is $(0,1] \ni t \mapsto x$, which does not terminate in $S$
although it is a geodesic with a smooth continuation
to $[0,1]$.

\section{Probability of One Event}
\label{one-event}

The probability of the occurrence of an event
$E$ given a state represented by 
$\psi \in \mathbb{S}(\mathcal{H})$ 
is given (see \cite{sbs}) by
$$
{\rm P}_{\psi} (E) := || E \psi ||^{2}. 
$$
It is straightforward (see \cite{sbs}) 
to show that the
formula  on the right side defining this probability defines a 
function of the structures 
$\pi'(E)$ and $\pi(\psi)$ 
of $\mathbb{CP} (\mathcal{H})$. 
In other words, for any point 
$x \in \mathbb{CP} (\mathcal{H})$
and complex projective subspace 
$S \subset \mathbb{CP} (\mathcal{H})$ we define 
the probability 
$ {\rm P}_{x} (S) := {\rm P}_{\psi} (E) $, where 
$\pi (\psi) = x$ and 
$S =\pi' (E)$.

It now remains to find a formula
in terms of the structures 
of $\mathbb{CP} (\mathcal{H})$ for this function. 
So we want to find a formula 
for $ {\rm P}_{\psi} (E)$ in terms of $x$ and $S$. 

First, we consider the case $E \psi \ne 0$.
Then 
$\theta  = d (x,S) $ by the
Projection Theorem, where 
$$
\cos \theta =  
\dfrac{|\langle \psi, E\psi \rangle|}
{||\psi||\, || E\psi||} =
\dfrac{|\langle E\psi, E\psi \rangle|}
{|| E\psi||} =
\dfrac{|| E\psi ||^{2}}{|| E\psi||} = || E\psi||. 
$$ 
Here we are using the result that $E = E^{2}$ and that 
$E$ is self-adjoint with respect to the Hermitian inner product. 
It follows that 
\begin{equation}
	\label{one-event-probability}
	{\rm P}_{x} (S) = 
	{\rm P}_{\psi} (E) = || E \psi ||^{2} = \cos^{2} \theta = 
	 \cos^{2} d(x,S). 
\end{equation}
So this gives a formula for the probability 
on the left side in
terms of a geometric property of the projective
space on the right side. 

Next, for the case $E\psi=0$, we have 
$ {\rm P}_{\psi} (E) = || E \psi ||^{2} = 0 $. 
However, we also have $x \in S^{\perp}$ in this case. 
Therefore $d(x,S) = \pi/2$ which in turn implies
$\cos d(x,S) = 0$. 
Consequently, 
$ 	{\rm P}_{x} (S) = \cos^{2} d(x,S)$
holds in this case, too. 

We have proved the following.
\begin{theorem}
	The single event probability ${\rm P}_{\psi} (E)$ in terms of
	the Hilbert space structures 
	$\psi \in \mathbb{S(\mathcal{H})}$ and an event
	$E \in \mathcal{L}(\mathcal{H})$ passes down to the
	probability in the complex projective space 
	${\rm P}_{x} (S) =  \cos^{2} d(x,S)$ of a point
	$x \in  \mathbb{CP} (\mathcal{H})$ and a complex
	projective subspace $S$, where $x = \pi_{2}(\psi)$ and 
	$S = \pi'(E)$. 
\end{theorem}

This theorem has the special case $S = \pi'(E)$ with 
$E = |\phi \rangle \langle \phi |$ for some 
$\phi \in \mathbb{S(\mathcal{H})}$. 
But in this case $S$ is the one point complex 
projective subspace 
whose point is 
$y:= \pi_{2} (\phi)$. 
Then we write $S = y$ instead of the correct notation
$S = \{ y \}$. 

We take $x = \pi_{2}(\psi)$ for some $\psi \in \mathbb{S(\mathcal{H})}$. 
Then for $ \langle \phi , \psi \rangle \ne 0$,
or equivalently $|\langle x, y \rangle| \ne 0$, 
we have 
$E \psi = |\phi \rangle \langle \phi | \psi =
\langle \phi , \psi \rangle \phi $ 
and therefore 
\begin{align*}
d(x,S) &= 
\langle \psi, E \psi \rangle 
=
\cos^{-1} 
\left( \dfrac{|\langle \psi, 
\langle \phi, \psi \rangle \phi \rangle|}
{|| \psi || \, || \langle \phi, \psi \rangle \phi ||} \right)
= 
\cos^{-1} \left(
\dfrac{|\langle \phi, \psi \rangle|^{2}}{|\langle \phi, \psi \rangle|}
\right) 
\\
&= 
\cos^{-1} (  |\langle \phi, \psi \rangle| ).
\end{align*}

Putting all this together we have 
the known result in \cite{schilling},
namely, 
$$ 
 |\langle \phi, \psi \rangle|^{2} = \cos^{2} d(x,S) = 
\cos^{2} d(x,y). 
$$
The other case $|\langle x, y \rangle| = 0$
is left to the reader.

\section{The consecutive probability of two events}

Suppose that $E$ and $F$ are events. 
The probability of $E$ then $F$ occurring in that time order, 
given a unit vector  $\psi$ representing a state, 
is defined as
$$
{\rm P}_{\psi} (E,F) := || FE \psi ||^{2}.
$$
This is called {\em (quantum) consecutive probability} in
\cite{sbs}. 
This is also known as Wigner's Rule. 
(See \cite{wigner}.)
Notice that the two events have a time order. 
Now in non-relativistic physics this is an absolute relation
valid for all observers. 
But in relativistic physics it means that the two events are
time-like and that time reversal is not taken to be a
symmetry. 

Next it is a question of how to write this probability in
purely geometric terms. 
We first analyze the case when $E \psi \ne 0$. 
Then putting $\phi := E \psi$  and 
$$
\phi' := \phi / || \phi || = \pi_{1} (\phi) = 
\pi_{1} ( E \psi )
$$
we have 
$$
{\rm P}_{\psi} (E,F) = 
|| FE \psi ||^{2} = || F \phi ||^{2} = 
|| F \Big( \dfrac{\phi}{|| \phi ||}  \Big) ||^{2} \cdot 
|| \phi ||^{2} =
{\rm P}_{\phi'} (F) \cdot {\rm P}_{\psi} (E)
$$
Now each of the factors on the right side can be
written in geometric terms by using the result of the
previous section. 
So this yields
$$
{\rm P}_{\psi} (E,F) = \cos^{2} (d (x', S')) \cos^{2} d(x,S), 
$$
where 
$\pi_{2} (\psi) = x$ and 
$S =\pi' (E)$ as before and also 
$\pi_{2} (\phi') = x'$ and 
$S' =\pi' (F)$. 

Next we have translate this to a definition of
${\rm P}_{x} (S, S') $ with no references to the Hilbert space. 
This comes down to showing how $x'$ is geometrically 
related to the other geometric structures. 

Now we see that
$$
x' = \pi_{2} (\phi') = 
\pi_{2} ( \pi_{1} ( E \psi ) ) = 
\pi_{3} ( E \psi ) = {\rm Pj} (x \downarrow S)
$$
where the last identity follows from the 
Projection Theorem. 
Putting all of this together gives the 
definition
$$
{\rm P}_{x} (S, S') := 
\cos^{2} 
\Big(d \big({\rm Pj} (x \downarrow S), S'\big) \Big) \cos^{2} d(x,S)
$$
This gives the desired result for the case $x \notin S^{\perp}$,
which is the condition $ E \psi \ne 0$ expressed in purely geometric terms. 
But it might be easier to read if we re-write
this definition as follows:
$$
{\rm P}_{x} (S, S') := 
\cos^{2}  d(y, S')  \,  \cos^{2} d(x,S) \qquad
{\rm where~} y = d \big({\rm Pj} (x \downarrow S), S'\big).
$$
Moreover, in this form it is now straightforward
to see how this generalizes to the consecutive
probability of any finite sequence of time
ordered events. 

For the case $E \psi = 0$, which is equivalent to
$x \in S^{\perp}$, we have
${\rm P}_{\psi} (E,F)  = 0$. 
Therefore, we define 
${\rm P}_{x} (S, S') := 0$ for all $S'$ in 
the case $x \in S^{\perp}$. 

Conditional probability is also important in
quantum probability theory. 
(See \cite{beltrametti} and \cite{sbs}). 
However, these are defined as quotients of
the consecutive and single event probabilities
and therefore have 
formulas as well in terms of the geometric properties
of the complex projective space.

\section{Accommodating von Neumann algebras}

The exposition so far has been implicitly given in the
context of $\mathcal{L}(\mathcal{H})$, a factor of Type~I. 
However, as we shall now see, this formalism goes through for any
von Neumann algebra which is given as a sub-algebra of some 
$\mathcal{L}(\mathcal{H})$. 
This is an important consideration, 
since Type~II von Neumann algebras arise
in quantum statistical mechanics and 
Type~III von Neumann algebras arise in quantum field theory. 
And after all, the foundations of quantum 
theory should include all possible cases, 
not for example just 
$\mathcal{L}(\mathcal{H})$ with $\mathcal{H}$
finite dimensional, where extra properties hold 
such as all operators being trace class. 

Now the essential point is that the kinematics
in this approach to quantum theory involves
just these three mathematical structures:
events, states and probabilities. 
The events need not be all the events in 
$\mathcal{L}(\mathcal{H})$ but only those
in the appropriate von Neumann sub-algebra.
Similarly, the states need not be all pure
states in 
$\mathcal{L}(\mathcal{H})$, but only those
that are relevant to the system being studied. 
These structures than pass down to the complex
projective space.
If density matrices are required, then these 
can also be passed down. 
(See \cite{sbs}.)
Eventually, one has a lattice of complex 
projective subspaces together with pure and
mixed states represented in the complex
projective space. 
In general, this will be a proper subset 
of the full lattice of complex subspaces
and a proper subset of states. 

The explicit choice of these structures
in the complex projective space is the
{\em quantization} of the physical system 
being studied. 
It follows that quantization is partly a
mathematical problem as well as partly a
physical problem. 
Quantization can be motivated by a classical
mechanical model of the system, 
but that is not necessary. 

Finally, probability must be defined in this
setting without any reference to the 
overlying Hilbert space. 
That is precisely the achievement of 
the present paper.  

It is important to note that 
in this approach the von Neumann algebra 
chosen need not be
a factor, even though physical 
considerations may lead one 
to prefer that in particular cases. 

\section{Concluding Remarks}
This approach is part of a general program to 
present various theories of physics in terms of
geometric structures. 
Thus, classical (conservative) mechanics is viewed in terms 
of symplectic geometry, special relativity in terms of
Minkowski geometry, 
electrodynamics in terms of a
connection in Minkowski geometry, 
general relativity in terms of differential geometry and
thermodynamics in terms of contact geometry. 
Here I am advocating a further 
extension of the known
way of viewing quantum theory in terms of 
complex projective geometry. 
In particular all the probability 
rules of quantum theory
are viewed in terms of the Riemannian metric structure,
which is only a part of the K\"ahler structure of 
complex projective geometry. 
This was already shown in the special case 
considered in \cite{ashtekar-schilling} and \cite{schilling}. 
Notice that the Hilbert space enters in this approach only as
a K\"ahler space which passes this structure to its
associated complex projective space. 
The vector space (also called linear) structure of 
the Hilbert space is not used. 

On possible avenue for further research is to find 
other K\"ahler spaces (or some generalization of them)
which admit theories similar to quantum theory. 
It might be interesting to understand how the
single `star' of quantum theory is distinguished in
the `galaxy' of all such generalizations. 

\vskip 0.6cm
\centerline{Acknowledgment}
\vskip 0.2cm 
I thank Ivo Terek for some very useful comments and for 
bringing references \cite{freed} and \cite{oneill}
to my attention.

\end{document}